\documentclass[runningheads,11pt]{llncs}

\usepackage{hyperref}

\usepackage[T1]{fontenc}
\usepackage[english]{babel}
\usepackage{amsmath,amssymb}
\usepackage[babel]{microtype}
\usepackage[dvipsnames]{xcolor}
\usepackage{bookmark}
\usepackage{verbatim}
\usepackage{subcaption}
\usepackage[export]{adjustbox}
\usepackage[noend]{algpseudocode}
\usepackage{algorithm}
\usepackage{multicol}
\usepackage{standalone}
\usepackage{xcolor}
\usepackage{enumerate}   
\usepackage{geometry}
\usepackage{lineno}
\usepackage[normalem]{ulem}
\usepackage{doi}

\usepackage{eqparbox}
\algnewcommand\LeftComment[1]{%
$\triangleright$ \eqparbox{COMMENT}{#1} \hfill %
}

\newtheorem{claim2}[theorem]{Claim}

\usepackage{mathtools}

\newcommand{\oldqed}{}
\def\endofFact{$\Diamond$}
\newenvironment{claimproof}{
  \renewcommand{\oldqed}{\qed}
  \renewcommand{\qed}{\endofFact}
  \begin{proof}
  \leftskip15pt\relax
}{
  \hfill\qed\end{proof}
  \renewcommand{\qed}{\oldqed}
} 

\newcommand{\steinercycle}{{\sc Steiner Multicycle}}
\newcommand{\snd}{{\sc Survivable Network Design}}
\newcommand{\tsp}{{\sc Traveling Salesman}}
\newcommand{\opt}{\mathrm{opt}}
\newcommand{\calC}{\mathcal{C}}
\newcommand{\calF}{\mathcal{F}}
\newcommand{\calT}{\mathcal{T}}
\newcommand{\Oh}{\mathrm{O}}

\newcommand{\defi}[1]{%
  \emph{\color{red!50!black}#1}%
}

\begin{document}

\title{Approximations for the Steiner Multicycle Problem}

\author{Cristina G. Fernandes\inst{1}\orcidID{0000-0002-5259-2859} \and
Carla N. Lintzmayer\inst{2}\orcidID{0000-0003-0602-6298} \and
Phablo~F.~S.~Moura\inst{3,4}\orcidID{0000-0002-8176-0874}}
\authorrunning{C. G. Fernandes, C. N. Lintzmayer, P. F. S. Moura}

\institute{Department of Computer Science. University of
S{\~a}o Paulo. Brazil\\
\email{\href{mailto:cris@ime.usp.br}{cris@ime.usp.br}}
\and
Center for Mathematics, Computing and Cognition. Federal
University of ABC. Santo Andr{\'e}, S{\~a}o Paulo, Brazil \\
\email{\href{mailto:carla.negri@ufabc.edu.br}{carla.negri@ufabc.edu.br}}
\and
Computer Science Department. Federal University of Minas
Gerais. Brazil \\
\and
Research Center for Operations Research \& Statistics. KU Leuven. Belgium\\
\email{\href{mailto:phablo.moura@kuleuven.be}{phablo.moura@kuleuven.be}}
}

\maketitle

\begin{abstract}
  The \steinercycle{} problem consists of, given a complete graph, a weight
  function on its vertices, and a collection of pairwise disjoint non-unitary
  sets called terminal sets, finding a minimum weight collection of
  vertex-disjoint cycles in the graph such that, for every terminal set, all of
  its vertices are in a same cycle of the collection.
  This problem generalizes the \textsc{Traveling Salesman} problem and
  therefore is hard to approximate in general.
  On the practical side, it models a collaborative less-than-truckload problem
  with pickup and delivery locations.
  Using an algorithm for the \textsc{Survivable Network Design} problem and
  $T$-joins, we obtain a 3-approximation for the metric case, improving on the
  previous best 4-approximation.
  Furthermore, we present an (11/9)-approximation for the particular case of
  the \steinercycle{} in which each edge weight is~1 or~2.
  This algorithm can be adapted to obtain a (7/6)-approximation when every
  terminal set contains at least~4 vertices.
  Finally, we devise an $\Oh(\lg n)$-approximation algorithm for the asymmetric
  version of the problem.

  \keywords{Combinatorial optimization \and Approximation algorithms \and Steiner
  problems \and Traveling salesman problem \and Collaborative logistics}
\end{abstract}

%%%%%%%%%%%%%%%%%%%%%%%%%%%%%%%%%%%%%%%%%%%%%%%%%%%%%%%%%%%%%%%%%%%%%%%%%%%%%%%%%%%%%%%%%%%%
\section{Introduction}

In the \steinercycle{} problem, one is given a complete graph~$G$, a weight
function $w \colon E(G) \to \mathbb{Q}_+$, and a collection~$\calT\subseteq
\mathcal{P}(V(G))$ of pairwise disjoint non-unitary sets of vertices, called
\defi{terminal sets}.
We say that a cycle~$C$ \defi{respects}~$\calT$ if, for all $T \in \calT$,
either every vertex of~$T$ is in~$C$ or no vertex of~$T$ is in~$C$, and a
set~$\calC$ of vertex-disjoint cycles \defi{respects} $\calT$ if all cycles
in~$\calC$ respect $\calT$ and every vertex in a terminal set is in some cycle
of~$\calC$. 
The cost of such set~$\calC$ is the sum of the edge weights over all cycles
in~$\calC$, a value naturally denoted by $w(\calC)$.
The goal of the \steinercycle{} problem is to find a set of vertex-disjoint
cycles of minimum cost that respects~$\calT$.
We denote by $\opt(G,w,\calT)$ the cost of such a minimum cost set.
Note that the number of cycles in a solution might be smaller than $|\calT|$,
that is, it might be cheaper to join some terminal sets in the same cycle.

We consider that, in a graph~$G$, a cycle is a non-empty connected subgraph
of~$G$ all of whose vertices have degree two.
Consequently, such cycles have at least three vertices.
Here, as a set $T \in \calT$ can have only two vertices, we would like to
consider a single edge as a cycle, of length two, whose cost is twice the
weight of the edge, so that the problem also includes solutions that choose to
connect some set from~$\calT$ with two vertices through such a length-2 cycle.
So, for each set $T \in \calT$ with $|T| = 2$, we duplicate in~$G$ the edge
linking the vertices in~$T$, and allow the solution to contain length-2 cycles.

The \steinercycle\ problem is a generalization of the \tsp{} problem (TSP),
thus it is NP-hard and its general form admits the same inapproximability
results as the TSP. 
It was proposed by Pereira \textit{et al.}~\cite{2018-pereira-etal} as a
generalization of the so-called \textsc{Steiner Cycle} problem (see
Salazar-Gonz{\'a}lez~\cite{Salazar2003}), with the assumption that the graph is
complete and the weight function satisfies the triangle inequality.
We refer to such an instance of the \steinercycle\ problem as \defi{metric},
and to the problem restricted to such instances as the \textsc{Metric}
\steinercycle\ problem.

Pereira \textit{et al.}~\cite{2018-pereira-etal} presented a 4-approximation
algorithm for the \textsc{Metric} \steinercycle\ problem, designed Refinement
Search and GRASP-based heuristics, and proposed an integer linear programming
formulation for the problem.
Lintzmayer \textit{et al.}~\cite{2020-lintzmayer-etal} then considered the
version restricted to the Euclidean plane and presented a randomized
approximation scheme for it, which combines some techniques for the
\textsc{Euclidean TSP}~\cite{1998-arora} and for the \textsc{Euclidean Steiner
Forest}~\cite{2015-borradaile-etal}.

On the practical side, the \steinercycle{} problem models a collaborative
less-than-truckload problem with pickup and delivery locations.
In this scenario, several companies operating in the same geographic regions
must periodically transport products between different locations. 
To reduce the costs of transporting their goods, these companies may
collaborate to create routes for shared cargo vehicles that visit the places
defined by them for the collection and delivery of their products (see Ergun
\textit{et al.}~\cite{Ergun2007a,Ergun2007b}).

This paper addresses three variations of the \steinercycle{}.
The first is the metric case, for which we present a 3-approximation, improving
on the previously best known.
The proposed algorithm uses an approximate solution~$S$ for a derived 
instance of the \snd{}\ problem and a minimum weight $T$-join in~$S$, 
where~$T$ is the set of odd-degree vertices in~$S$.
The second one is the so-called $\{1,2\}$-\steinercycle{} problem, in which the
weight of each edge is either~1 or~2.
Note that this is a particular case of the metric one, and it is a
generalization of the $\{1,2\}$-TSP, therefore it is also
APX-hard~\cite{1993-papadimitriou-yannakakis}.
In some applications, there might be little information on the actual cost of
the connections between points, but there might be at least some distinction
between cheap connections and expensive ones.
These situations could be modeled as instances of the
$\{1,2\}$-\steinercycle{}.
For this variation, we design an $\frac{11}{9}$-approximation following the
strategy for the $\{1,2\}$-TSP proposed by Papadimitrou and
Yannakakis~\cite{1993-papadimitriou-yannakakis}.
The third variation is the asymmetric case, in which one is now given a
complete digraph in which the weight of an arc $(u,v)$ is not necessarily the
same as the weight of the arc $(v,u)$, but the weights still satisfy the
triangle inequality.
For this case, we design an $\Oh(\lg n)$-approximation algorithm, where~$n$ is
the number of vertices in the graph, following some ideas for the Asymmetric
TSP proposed by Frieze, Galbiati, and Maffioli~\cite{Frieze1982}.

Note that the three variations we consider are metric.
In this case, we assume that the terminal sets partition the vertex set. 
Indeed, because the graph (or digraph) is complete and the weight function is
metric, any solution containing non-terminal vertices does not have its cost
increased by shortcutting these vertices (that is, removing them and adding the
edge linking their neighbors in the cycle).
Therefore, the set of cycles of any solution is a \defi{2-factor that respects}
the terminal sets.
A \defi{2-factor} is a set of vertex-disjoint cycles that spans all vertices of
the graph.

A preliminary version of this paper was published in the 
LATIN 2022 proceedings~\cite{2022-fernandes-etal}.
In addition to the results presented there, this manuscript
contains a new algorithm for the asymmetric version of the problem, improved
proofs, more examples, and a detailed discussion on minimum weight
triangle-free 2-factors.

The 3-approximation for the \textsc{Metric} \steinercycle\ is presented in
Section~\ref{sec:metric}, together with a discussion involving the previous
4-approximation and the use of perfect matchings on the set of odd degree
vertices of intermediate structures. 
The $\{1,2\}$-\steinercycle{} problem is addressed in Section~\ref{sec:onetwo}.
The asymmetric case is investigated in Section~\ref{sec:asymmetric}, 
and we make some final considerations in Section~\ref{sec:final}.

\section{Metric Steiner Multicycle problem}
\label{sec:metric}

An instance for the \steinercycle{} is also an instance for the well-known
\textsc{Steiner Forest} problem~\cite[Chapter~22]{Vazirani2002}, but the goal
in the latter is to find a minimum weight forest in the graph that connects
vertices in the same terminal set, that is, every terminal set is in some
connected component of the forest.
The optimum value of the \textsc{Steiner Forest} is a lower bound on the
optimum for the \steinercycle{}: one can  produce a feasible solution for the
\textsc{Steiner Forest} from an optimal solution for the \steinercycle{} by
throwing away one edge in each cycle without increasing its cost.

The existing 4-approximation~\cite{2018-pereira-etal} for the metric
\steinercycle\ problem is inspired by the famous 2-approximation for the metric
TSP~\cite{RosenkrantzSL77}, and consists in doubling the edges in a Steiner
forest for the terminal sets and shortcutting an Eulerian tour in each of its
components to a cycle.
As there are 2-approximations for the \textsc{Steiner Forest} problem, this
leads to a 4-approximation. 

It is tempting to try to use a perfect matching on the odd-degree vertices of
the approximate Steiner forest solution, as Christofides'
algorithm~\cite{Christofides76} does to achieve a better ratio for the
\textsc{Metric TSP}.  
However, the best upper bound we can prove so far on such a matching is the
weight of the approximate Steiner forest solution, which implies that such a
matching weights at most twice the optimum.
With this bound, we also derive a ratio of at most~4.

Another problem that can be used with this approach is known as the \snd\
problem~\cite[Chapter~23]{Vazirani2002}.
An instance of this problem consists of the following: a graph~$G$, a weight
function $w \colon E(G) \to \mathbb{Q}_+$, and a non-negative integer~$r_{ij}$
for each pair of vertices $i,j$ with $i \neq j$, representing a connectivity
requirement.  
The goal is to find a minimum weight subgraph~$G'$ of~$G$ such that, for every
pair of vertices $i, j \in V(G)$ with $i \neq j$, there are at least $r_{ij}$
edge-disjoint paths between~$i$ and~$j$ in~$G'$. 

From an instance of the \steinercycle{} problem, we can naturally define an
instance of the \snd\ problem:  set $r_{ij} = 2$ for every two vertices $i,j$
in the same terminal set, and set $r_{ij} = 0$ otherwise. 
As all vertices are terminals, all connectivity requirements are defined in
this way.
The optimum value of the \snd\ problem is also a lower bound on the optimum for
the \steinercycle{} problem: indeed an optimal solution for the \steinercycle{}
problem is a feasible solution for the \snd\ problem with the same cost.

There also exists a 2-approximation for the \snd\ problem~\cite{Jain2001}.
By applying the same approach of the 2-approximation for the metric TSP, of
doubling edges and shortcutting, we achieve a ratio of~4 for the metric
\steinercycle{} again.
However, next, we will show that one can obtain a 3-approximation for the metric
\steinercycle{} problem, from a 2-approximate solution for the \snd\ problem,
using not a perfect matching on the odd degree vertices of such solution, but
the related concept of $T$-joins.

\subsection{A 3-approximation algorithm for the metric case}

Let~$T$ be a set of vertices of even size in a graph~$G$. 
A set~$J$ of edges in~$G$ is a \defi{$T$-join} if the collection of vertices
of~$G$ that are incident to an odd number of edges in~$J$ is exactly~$T$. 
Any perfect matching on the vertices of~$T$ is a $T$-join, so $T$-joins are, in
some sense, a generalization of perfect matching on a set~$T$. 
It is known that a $T$-join exists in~$G$ if and only if the number of vertices
from~$T$ in each component of~$G$ is even. 
Moreover, there are polynomial-time algorithms that, given a connected
graph~$G$, a weight function $w \colon E(G) \to \mathbb{Q}_+$, and an even
set~$T$ of vertices of~$G$, find a minimum weight $T$-join in~$G$.
For these and more results on $T$-joins, we refer the reader to the book by
Schrijver~\cite[Chapter~29]{Schrijver2003}.

The idea of our 3-approximation is similar to
Christofides~\cite{Christofides76}. 
It is presented in Algorithm~\ref{alg:metric:3approx}.
Let $(G,w,\calT)$ be a metric instance of the \steinercycle\ problem.
The first step is to build the corresponding \snd\ problem instance and to
obtain a 2-approximate solution~$G'$ for this instance. 
The procedure \textsc{2ApproxSND} represents the algorithm by Jain~\cite{Jain2001}
for the \snd.
The second step considers the set~$T$ of the vertices in~$G'$ of odd degree and
finds a minimum weight $T$-join~$J$ in~$G'$. 
The procedure \textsc{MinimumTJoin} represents the algorithm by Edmonds and
Johnson~\cite{EdmondsJ1973} for this task. 
Finally, the Eulerian graph $H$ obtained from~$G'$ by doubling the edges in~$J$
is built and, by shortcutting an Eulerian tour for each component of~$H$, one
obtains a 2-factor~$\calC$ in~$G$ that is the output of the algorithm.
The procedure \textsc{Shortcut} represents this part in
Algorithm~\ref{alg:metric:3approx}.

\begin{algorithm}
  \begin{algorithmic}[1]
    \Require{a complete graph $G$, a weight function $w \colon E(G) \to
    \mathbb{Q}_+$ satisfying the triangle inequality, and a partition $\calT =
    \{T_1,\ldots,T_k\}$ of $V(G)$}
    \Ensure{a 2-factor~$\calC$ in $G$ that respects~$\calT$}
    
    \State $r_{ij} \gets 2$ for every $i,j \in T_a$ for some $1 \leq a \leq k$ 
    \State $r_{ij} \gets 0$ for every $i \in T_a$ and $j \in T_b$ for $1 \leq a < b \leq k$
    \State $G' \gets$ \Call{2ApproxSND}{$G$, $w$, $r$}
    \State Let $T$ be the set of odd-degree vertices in $G'$
    \State Let $w'$ be the restriction of $w$ to the edges in $G'$
    \State $J \gets$ \Call{MinimumTJoin}{$G'$, $w'$, $T$} \label{line:metric:3approx:Tjoin}
    \State $H \gets G'+J$
    \State $\calC \gets$ \Call{Shortcut}{$H$}
    \State \Return $\calC$
  \end{algorithmic}
  \caption{\textsc{SteinerMulticycleApprox\_Metric}($G$, $w$, $\calT$)}
  \label{alg:metric:3approx}
\end{algorithm}

Because the number of vertices of odd-degree in any connected graph is even, 
the number of vertices with odd degree in each component of~$G'$ is even.
Therefore there is a $T$-join in~$G'$.
Moreover, the collection~$\calC$ produced by Algorithm~\ref{alg:metric:3approx} is
indeed a feasible solution for the \steinercycle.

Next, we prove that the proposed algorithm is a 3-approximation. 

\begin{theorem}
\label{thm:metric:3approx}
  Algorithm~\ref{alg:metric:3approx} is a 3-approximation for the
  \textsc{Metric} \steinercycle{} problem.
\end{theorem}
\begin{proof}
  First, it suffices to prove that $w(J) \leq \frac12w(G')$.
  Indeed, because~$G'$ is a 2-approximate solution for the \snd\ problem, and
  the optimum for this problem is a lower bound on $\opt(G,w,\calT)$, we have
  that $w(G') \leq 2\,\opt(G,w,\calT)$.
  Hence we deduce that $w(J) \leq \opt(G,w,\calT)$, and therefore that
  $w(\calC) \leq w(G') + w(J) \leq 3\,\opt(G,w,\calT)$.
  We now show that inequality $w(J) \leq \frac12w(G')$ holds. 

  A \defi{bridge} is an edge~$uv$ in a graph whose removal leaves~$u$ and~$v$
  in different components of the resulting graph.  
  First, observe that we can delete from~$G'$ any bridges and the remaining
  graph, which we still call~$G'$, remains a solution for the \snd\ problem
  instance.
  Indeed a bridge is not enough to assure the connectivity requirement between
  two vertices in the same terminal set, so it will not separate any such pair
  of vertices, and hence it can be removed.
  In other words, we may assume that each component of~$G'$ is
  2-edge-connected. 

  Edmonds and Johnson~\cite{EdmondsJ1973} gave an exact description of a
  polyhedra related to $T$-joins.
  This description will help us to prove the claim.
  For a set~$S$ of edges in a graph $(V,E)$, let~$v(S)$ denote the
  corresponding $|E|$-dimensional incidence vector (with~1 in the $i$-th
  coordinate if edge~$i$ lies in~$S$ and 0 otherwise). 
  For a set~$X$ of vertices, let $\delta(X)$ denote the set of edges with one
  endpoint in~$X$ and the other in~$V \setminus X$. 
  An \defi{upper $T$-join} is any superset of a $T$-join.
  Let~$P(G,T)$ be the convex hull of all vectors $v(J)$ corresponding to the
  incidence vector of upper $T$-joins~$J$ of a graph~$G=(V,E)$.
  The set~$P(G,T)$ is called the \defi{up-polyhedra of $T$-joins}, and it is
  described by
  \begin{eqnarray}
   \sum_{e \in \delta(W)} x(e) \geq 1 & & 
     \mbox{for every $W\subseteq V$ such that $|W \cap T|$ is odd}, \label{polyhedra:eq1} \\
      0 \leq x(e) \leq 1 & & \mbox{for every edge $e \in E$. \label{polyhedra:eq2}}
  \end{eqnarray}
  (For more on this, see~\cite[Chapter~29]{Schrijver2003}.)

  So, as observed in~\cite{BansalBT2009}, any feasible solution~$x$ to the
  system of inequalities above can be written as a convex combination of upper
  $T$-joins, that is, $x = \sum \alpha_i\,v(J_i)$, where $0 \leq \alpha_i \leq
  1$ and $\sum_i \alpha_i = 1$, leading to the following. 

  \begin{corollary}[Corollary~1 in~\cite{BansalBT2009}]
  \label{cor:dosoutros}
    If all the weights $w(e)$ are non-negative, then, given any feasible
    assignment $x(e)$ satisfying the inequalities above, there exists a
    $T$-join with weight at most $\sum_{e \in E} w(e)x(e)$.
  \end{corollary}

  Recall that, for each component~$C$ of $G'$, $|V(C) \cap T|$ is even.
  Hence, for every~$W \subseteq V(G')$ such that $|W \cap T|$ is odd, there must
  exist a component $C$ of~$G'$ with~$V(C) \cap W \neq \emptyset$, and
  $V(C) \setminus W \neq \emptyset$.
  As a consequence, it holds that~$|\delta(W)|\ge 2$ because every component
  of~$G'$ is 2-edge-connected.
  Consider now the $|E(G')|$-dimensional vector~$\bar x$ which assigns value
  $1/2$ to each edge of~$G'$.
  From the discussion above, it is clear that~$\bar x$ satisfies
  inequalities~\eqref{polyhedra:eq1} and~\eqref{polyhedra:eq2} for~$G'$
  and~$T$.
  Then Corollary~\ref{cor:dosoutros} guarantees that there is a $T$-join~$J$
  in~$G'$ such that $w(J) \leq \frac12\,w(G')$. 
  This completes the proof of the theorem.
\hfill$\qed$
\end{proof}

\subsection{Matchings, $T$-joins, and Steiner forests}

Because~$G$ is complete and~$w$ is metric, the proof of
Theorem~\ref{thm:metric:3approx} in fact implies that a minimum weight perfect
matching in the graph $G[T]$ weights at most $w(G')/2$, and therefore at most
$\opt(G,w,\calT)$. 
However, we have no direct proof for this fact; only this argument that goes
through a minimum weight $T$-join. 
But this fact means that one can exchange line~\ref{line:metric:3approx:Tjoin}
to compute, instead, a minimum weight perfect matching~$J$ in $G[T]$.

We investigated the possibility that one could achieve a ratio of~3 using a
Steiner forest instead of a survivable network design solution.
However, using a $T$-join does not work so well with the Steiner forest, once
its components are not 2-edge-connected.  
Indeed, if~$T$ is the set of odd-degree vertices in a Steiner forest~$F$, a
bound as in the proof of Theorem~\ref{thm:metric:3approx} on a minimum weight
$T$-join in~$F$ would not hold in general: there are examples for which such a
$T$-join in~$F$ has weight $w(F)$.  

\newcommand{\SND}{\mathrm{SND}}
\newcommand{\SF}{\mathrm{SF}}
\newcommand{\SMC}{\mathrm{SMC}}

In this paragraph, let $\opt_{\SND}$ denote the optimum value for the \snd\
instance used in Algorithm~\ref{alg:metric:3approx}, and $\opt_{\SF}$ denote the
optimum value for the \textsc{Steiner Forest} instance used in the
4-approximation from the literature~\cite{2018-pereira-etal}. 
Let $\opt_{\SMC}$ be the \steinercycle\ optimum value. 
Note that $\opt_{\SF} \leq \opt_{\SND} \leq \opt_{\SMC} \leq 2\,\opt_{\SF}$,
where the last inequality holds because a duplicated Steiner forest solution
leads to a cheaper feasible solution for the \snd\ and the \steinercycle\
instances.
Let~$G'$ and~$J$ be the subgraph and the $T$-join used in
Algorithm~\ref{alg:metric:3approx}, respectively, and let~$M$ be a minimum weight
perfect matching in $G[T]$. 
Then $w(M) \leq w(J) \leq \frac12w(G') \leq \opt_{\SND} \leq w(G')$.
(For the first inequality, recall that $J$ is a $T$-join in $G'$ while $M$ is a
minimum weight perfect matching in $G[T]$.)
If~$T'$ is the set of odd-degree vertices in an optimal Steiner forest and~$M'$
is a minimum weight perfect matching in~$G[T']$, then $w(M') \leq 2\opt_{\SF}$,
and there are instances for which this upper bound is tight. 
So, as far as we know, there might be an instance where $w(M') > \opt_{\SMC}$. 
Even if this is not the case, in fact, what we can compute in polynomial time
is a minimum weight perfect matching~$M''$ for the set of odd-degree vertices
in a 2-approximate Steiner forest solution, so it would still be possible that
$w(M'') > \opt_{\SMC}$ for some instances.  
We tried to find an instance where this is the case, but we have not succeeded
so far.

\section{$\{1,2\}$-Steiner Multicycle problem}
\label{sec:onetwo}

In this section, we will address the particular case of the metric
\steinercycle\ problem that allows only edge weights~1 or~2. 

It is a well-known result that there exists a polynomial-time algorithm for
finding a 2-factor of minimum weight in weighted
graphs~\cite{LovaszP1986,Tutte1954}.
Specifically, for a complete graph on~$n$ vertices, one can find such a
2-factor by finding a maximum weight perfect matching in a graph with
$\Oh(n^2)$ vertices and edges.
This can be done in time $\Oh(n^4)$ using Orlin's maximum flow
algorithm~\cite{Orlin2013}.

The algorithm for this case of the \steinercycle\ problem starts from a minimum
weight 2-factor of the given weighted graph, and then repeatedly joins two
cycles until a feasible solution is obtained.
The key to guaranteeing a good approximation ratio is a clever choice of the
cycles to join at each step.
To proceed with the details, we need the following definitions.

Let $(G,w,\calT)$ be an instance of the \steinercycle\ problem with $w \colon
E(G) \to \{1,2\}$.
Recall that $\bigcup_{T \in \calT} T = V(G)$, and that, for each set $T \in
\calT$ with $|T| = 2$, we duplicated in~$G$ the edge linking the vertices
in~$T$, to allow the solution to contain length-2 cycles. 
We say an edge~$e \in E(G)$ is an \defi{$i$-edge} if $w(e) = i$, for $i \in
\{1,2\}$.
A cycle containing only 1-edges is called \defi{pure}; otherwise, it is called
\defi{nonpure}.

All steps of the procedure are summarized in Algorithm~\ref{alg:onetwo:119approx}. 
In what follows, we explain some auxiliary procedures used in the algorithm. 

\begin{algorithm}
  \begin{algorithmic}[1]
    \Require{a complete graph $G$, a weight function $w \colon E(G) \to
    \{1,2\}$, and a partition $\calT = \{T_1,\ldots,T_k\}$ of $V(G)$}
    \Ensure{a 2-factor~$\calC$ in~$G$ that respects~$\calT$}
    
    \State $\calF \gets$ \Call{Special2Factor}{$G$, $w$}
    \State $B \gets $ \Call{BuildBipartiteGraph}{$G$,$w$,$\calT$,$\calF$}\label{line:onetwo:119approx:Bgraph}
    \State $M \gets$ \Call{MaximumMatching}{$B$}\label{line:onetwo:119approx:maxmatching}
    \State Let~$D$ be a digraph such that $V(D) = \calF$ and there is an arc
    $(C,C') \in E(D)$ if~$C$ is matched by~$M$ to a vertex of~$C'$
    \State $D' \gets$ \Call{SpecialSpanningGraph}{$D$}
    \State $\calC' \gets$ \Call{JoinComponentCycles}{$\calF$, $D'$} (see Section~\ref{sub:sec:onetwo:joiningcycles})
    \State $\calC \gets$ \Call{JoinDisrespectingCycles}{$\calC'$, $D'$, $\calT$} (see Section~\ref{sub:sec:onetwo:joiningcycles})
    \State \Return $\calC$
  \end{algorithmic}
  \caption{\textsc{SteinerMulticycleApprox\_12Weights}($G$, $w$, $\calT$)}
  \label{alg:onetwo:119approx}
\end{algorithm}

Procedure \Call{Special2Factor}{} finds a minimum weight
2-factor~$\calF$ of $(G,w)$ with the two following properties:
\begin{enumerate}[(i)]
    \item\label{onetwo:2-factor:p1} $\calF$ contains at most one nonpure cycle; and 
    \item\label{onetwo:2-factor:p2} if $\calF$ contains a nonpure cycle, no
    1-edge in~$G$ connects an endpoint of a 2-edge in the nonpure cycle to a
    pure cycle in~$\calF$.
\end{enumerate}
Given any minimum weight 2-factor~$\calF'$, one can construct in polynomial
time a 2-factor~$\calF$ from~$\calF'$ having
properties~\eqref{onetwo:2-factor:p1} and~\eqref{onetwo:2-factor:p2} as
follows.
To ensure property~\eqref{onetwo:2-factor:p1}, recall that the graph is
complete, so we repeatedly join two nonpure cycles by removing one 2-edge from
each and adding two appropriate edges that turn them into one cycle.  
This clearly does not increase the weight of the 2-factor and reduces the
number of cycles.
To ensure property~\eqref{onetwo:2-factor:p2}, while there is a 1-edge~$yz$
in~$G$ connecting a 2-edge $xy$ of the nonpure cycle to a 1-edge $wz$ of a pure
cycle, we remove~$xy$ and~$wz$ and add~$yz$ and~$xw$, reducing the number of
cycles without increasing the weight of the 2-factor.
The resulting 2-factor is returned by \Call{Special2Factor}{}.

In order to modify~$\calF$ into a 2-factor that respects~$\calT$,
without increasing too much its weight, Algorithm~\ref{alg:onetwo:119approx} builds
some auxiliary structures that capture how the cycles in~$\calF$ attach to
each other.

The second step of Algorithm~\ref{alg:onetwo:119approx} is to build a bipartite
graph~$B$ (line~\ref{line:onetwo:119approx:Bgraph}) as follows.
Let $V(B) = V(G) \cup \{C \in \calF \colon C \text{ is a pure cycle}\}$ and there
is an edge $vC$ in $E(B)$ if~(i) $v \not\in V(C)$ and~$C$ does not respect
$\calT$, and (ii) there is a vertex $u \in V(C)$ such that $uv$ is a 1-edge.
Note that the only length-2 cycles in~$G$, and thus in~$\calF$, are those 
connecting a terminal set of size~$2$.
So such cycles respect~$\calT$ and, hence, if they are in~$B$ (that is, 
if they are pure), they are isolated vertices in~$B$. 
Procedure \Call{MaximumMatching}{} in
line~\ref{line:onetwo:119approx:maxmatching} computes in polynomial time a
maximum matching~$M$ in~$B$ (e.g., using Edmonds'
algorithm~\cite{Edmonds1965}).

Algorithm~\ref{alg:onetwo:119approx} then proceeds by building a digraph~$D$
where $V(D) = \calF$ and there is an arc $(C,C') \in E(D)$ if~$C$ is matched
by~$M$ to a vertex of~$C'$.
Note that the vertices of~$D$ have outdegree~$0$ or~$1$, and the cycles in~$B$
unmatched by~$M$ have outdegree~$0$ in~$D$.
In particular, all pure length-2 cycles in $F$ have outdegree~$0$ in~$D$,
because they are isolated in~$B$, and therefore unmatched.
If there is a nonpure cycle in~$\calF$, it also has outdegree~$0$ in~$D$.
Therefore, any length-2 cycle in~$\calF$, pure or nonpure, has outdegree~0
in~$D$.
However, these vertices with outdegree~0 in~$D$ might have an indegree
different from~0.
Next, Algorithm~\ref{alg:onetwo:119approx} applies procedure
\Call{SpecialSpanningGraph}{$D$} to find a spanning digraph~$D'$ of~$D$ whose
components are in-trees of depth~1, length-2 paths, or trivial components that
correspond to isolated vertices of~$D$.
This takes linear time and consists of a procedure described by
Papadimitrou and Yannakakis~\cite{1993-papadimitriou-yannakakis}, applied to
each nontrivial component of~$D$.
See Figure~\ref{fig:onetwo:119approx} for an example of these constructions.

\begin{figure}[ht!]
    \centering

    \begin{subfigure}{\textwidth}
        \centering
        \includegraphics[width=0.8\textwidth]{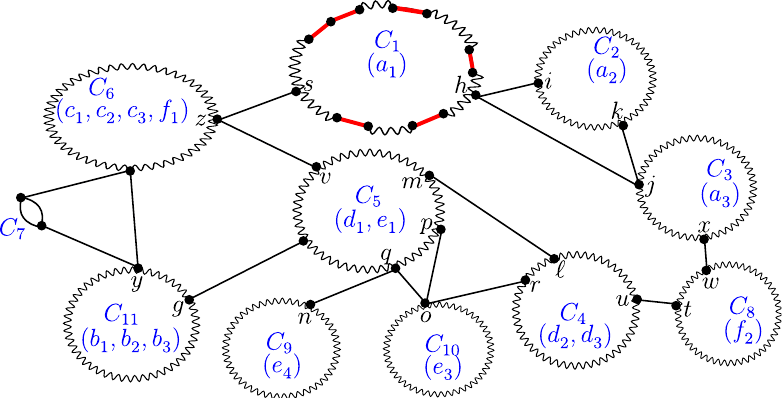}
        \caption{Original graph~$G$ and the 2-factor (depicting only 1-edges in
        black and straight lines and some 2-edges in bold and red lines;
        squiggly lines correspond to one or more 1-edges).
        Inside each $C_i$, in parenthesis, we list some of the terminal
        vertices it contains.}
        \label{fig:onetwo:119approx:twofactor}
    \end{subfigure}

    \vspace{3mm}

    \begin{subfigure}{0.5\textwidth}
        \centering
        \includegraphics[width=\textwidth]{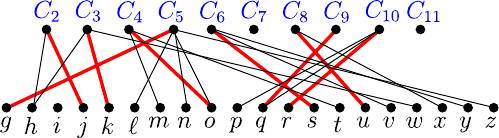}
        \caption{Bipartite graph~$B$ and a matching~$M$ highlighted in red and
        bold.
        Note that there is no edge incident to~$C_{11}$ because it already
        respects~$\calT$.}
        \label{fig:onetwo:119approx:matching}
    \end{subfigure}
    \hfill
    \begin{subfigure}{0.45\textwidth}
        \centering
        \includegraphics[width=0.65\textwidth]{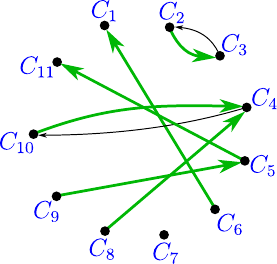}
        \caption{Digraph~$D$ and corresponding subgraph~$D'$ highlighted in
        green and bold.}
        \label{fig:onetwo:119approx:digraph}
    \end{subfigure}

    \caption{Auxiliar graphs and structures built by Algorithm~\ref{alg:onetwo:119approx}.}
    \label{fig:onetwo:119approx}
\end{figure}

At last, Algorithm~\ref{alg:onetwo:119approx} joins some cycles of~$\calF$ in
order to obtain a 2-factor that respects~$\calT$.
This will happen in two phases.
In the first phase, we join cycles that belong to the same component of~$D'$.
In the second (and last) phase, we repeatedly join cycles if they have vertices
from the same set in~$\calT$, to obtain a feasible solution to the problem.
This final step prioritizes joining cycles that have at least one 2-edge.

Details of these two phases, done by procedures \Call{JoinComponentCycles}{}
and \Call{JoinDisrespectingCycles}{}, as well as the analysis of the cost of
joining cycles, are given in Section~\ref{sub:sec:onetwo:joiningcycles}.
For now, observe that all cycles at the end of this process respect~$\calT$.
Also, note that length-2 cycles exist in the final solution only if they
initially existed in~$\calF$ and connected terminals of some set~$T \in \calT$
with~$|T|=2$.
The analysis of the approximation ratio of the algorithm is discussed 
in Section~\ref{sub:sec:onetwo:119approx}.

\subsection{Joining cycles}
\label{sub:sec:onetwo:joiningcycles}

In the first phase, we join cycles in~$\calF$ if they belong to the same
component of~$D'$, which can be either an in-tree of depth~1 or a length-2
path.

An in-tree of depth~1 of~$D'$ consists of a root~$C$ and some other cycles
$\{C_j\}_{j=1}^t$, with $t \geq 1$. 
Note that each arc $(C_j,C)$ can be associated with a 1-edge from~$G$ such that
no two edges are incident on the same vertex in~$C$, because they came from the
matching~$M$.
Also, note that if the nonpure cycle or a length-2 cycle appears in some
in-tree, it could only be the root~$C$.
Let~$v_j$ be the endpoint in $C$ of the edge associated with arc $(C_j,C)$, for
every $j \in \{1, \ldots, t\}$.
Rename the cycles $\{C_j\}_{j=1}^t$ so that, if we go through the vertices
of~$C$ in order, starting from $v_1$, these vertices appear in the order
$v_1,\ldots,v_t$.
We join all cycles in this in-tree into one single cycle in the following
manner.
For each~$v_i$ in~$C$, if~$v_{i+1}$ is adjacent to~$v_i$ in~$C$, then we
join~$C_i$ and~$C_{i+1}$ with~$C$ as in Figure~\ref{fig:onetwo:join-in-trees1}.
Otherwise, we join~$C$ and~$C_i$ as in Figure~\ref{fig:onetwo:join-in-trees2}.
We shall consider that the new cycle contains at least one 2-edge.

\begin{figure}[ht!]
    \centering

    \begin{subfigure}{\textwidth}
        \centering
        \includegraphics[width=0.8\textwidth]{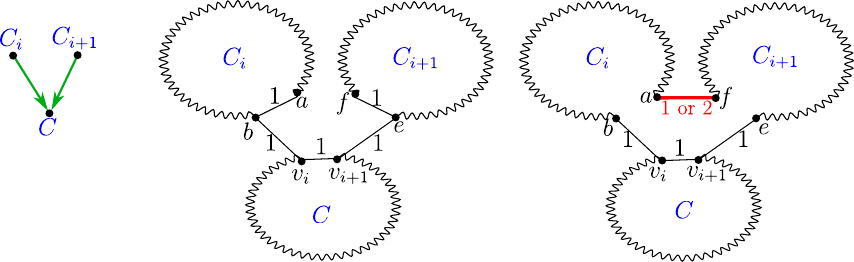}
        \caption{There are adjacent vertices $v_i$ and $v_{i+1}$ in~$C$.}
        \label{fig:onetwo:join-in-trees1}
    \end{subfigure}

    \vspace{3mm}

    \begin{subfigure}{\textwidth}
        \centering
        \includegraphics[width=0.8\textwidth]{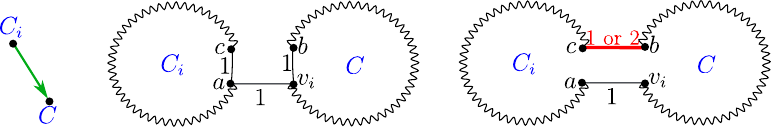}
        \caption{Vertex $v_i$ is not adjacent in~$C$ to another $v_j$.}
        \label{fig:onetwo:join-in-trees2}
    \end{subfigure}

    \caption{Joining cycles that belong to in-trees of~$D'$ into a unique cycle.
    The bold red edges will be considered as 2-edges even if they are 1-edges.}
    \label{fig:onetwo:join-in-trees}
\end{figure}

As for a component of~$D'$ which is a length-2 path, let~$C_i$, $C_j$,
and~$C_k$ be the three cycles that compose it, being~$C_i$ the beginning of the
path and~$C_k$ its end.
Note that if the nonpure cycle appears in some length-2 path, it could only
be~$C_k$.
The arcs $(C_i,C_j)$ and $(C_j,C_k)$ are also associated with 1-edges of~$G$,
but now it may be the case that such edges share their endpoint in~$C_j$.
If that is not the case, then we join these three cycles as shown in
Figure~\ref{fig:onetwo:join-paths1}.
Otherwise, we join the three cycles as shown in Figure~\ref{fig:onetwo:join-paths2}.
We shall also consider that the new cycle contains at least one 2-edge.

\begin{figure}[ht!]
    \centering

    \begin{subfigure}{\textwidth}
        \centering
        \includegraphics[width=\textwidth]{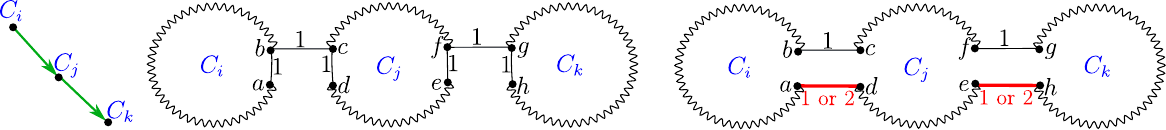}
        \caption{The edges do not share an endpoint.}
        \label{fig:onetwo:join-paths1}
    \end{subfigure}

    \vspace{3mm}

    \begin{subfigure}{\textwidth}
        \centering
        \includegraphics[width=\textwidth]{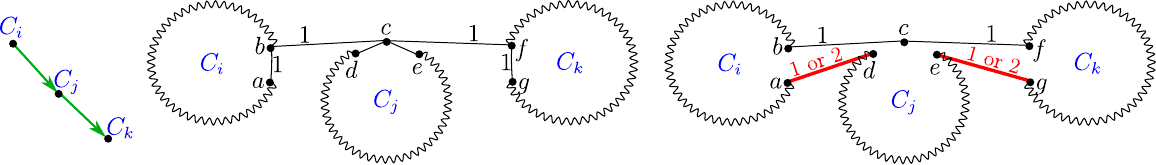}
        \caption{The edges share an endpoint.}
        \label{fig:onetwo:join-paths2}
    \end{subfigure}

    \caption{Joining cycles that belong to length-2 paths of~$D'$ into a unique
    cycle.
    The bold red edges will be considered as 2-edges even if they are 1-edges.}
    \label{fig:onetwo:join-paths}
\end{figure}

Let~$\calC'$ be the resulting 2-factor after the first phase.
This is the output of procedure \Call{JoinComponentCycles}{}. 
It may still be the case that two separated cycles in~$\calC'$ contain
terminals from the same set $T \in \calT$.
So, in the last phase, while there are two such cycles, join them in the
following order of priority: both cycles contain a 2-edge, precisely one of the
cycles contains a 2-edge, and none contains a 2-edge.
The resulting 2-factor of this phase, denoted by~$\calC$, is computed by
\Call{JoinDisrespectingCycles}{} and is the one returned by
Algorithm~\ref{alg:onetwo:119approx}.

\medskip

Now we proceed to analyze the cost increase caused by joining cycles in these
two phases.
Note that $w(\calC)$ is equal to $w(\calF)$ plus some value due to the
increases caused by joining cycles.

For the first phase, we charge the increment of the cost for joining cycles to
some of the vertices in the cycles being joined.
This is done in such a way that each vertex is charged at most once according
to the following.

\begin{claim2}
\label{clm:onetwo:phase1}
  Each vertex not incident to a 2-edge of $\calF$ is charged at most $2/9$
  during the first phase, and no other vertex is charged.
\end{claim2}
\begin{proof}
  Consider an in-tree of depth~1 with root~$C$ and cycles $C_1,\ldots,C_t$ with
  $t \geq 1$.
  When we join cycles~$C_i$ and~$C_{i+1}$ with~$C$, as in
  Figure~\ref{fig:onetwo:join-in-trees1}, note that the increase on the cost is at
  most~1.  
  We charge this cost to the vertices in~$C_i$ and~$C_{i+1}$, which are at
  least~6 (3 per cycle), thus costing at most~1/6 per vertex.  When we only
  join a cycle~$C_i$ with~$C$, as in Figure~\ref{fig:onetwo:join-in-trees2}, the
  increase is also at most~1. 
  We charge this cost to the vertices in~$C_i$ and also to the two vertices
  involved in~$C$.
  Since there are at least~3 vertices in~$C_i$, each of these vertices is
  charged at most~1/5.  
  Note that, indeed, each vertex is charged at most once. 
  Moreover, if~$C$ is the nonpure cycle, then, by
  property~\eqref{onetwo:2-factor:p2}, the edges in~$C$ incident to~$v_i$ and
  to the next vertex in~$C$ must be 1-edges.

  Consider now a length-2 path with vertices~$C_i$, $C_j$, and~$C_k$.
  The cycles $C_i$, $C_j$, and $C_k$ were joined as in
  Figures~\ref{fig:onetwo:join-paths1} and~\ref{fig:onetwo:join-paths2}, so the
  extra cost is at most~2, which is charged to the at least~9 vertices that
  belong to these cycles, giving a cost of at most~2/9 per vertex.
\hfill$\qed$
\end{proof}

As for the last phase, the increase in the cost will be considered for each
pair of cycles being joined.  
If both cycles contain 2-edges, joining them will not increase the cost of the
solution. 
If only one of the cycles contains a 2-edge, then the increase in the cost is
at most~1.
Joining cycles that do not contain 2-edges may increase the cost by~2.

\begin{claim2}
\label{clm:onetwo:phase2}
  The increase in the last phase is at most $c_p$, where $c_p$ is the number of
  pure cycles in~$F$ that do not respect~$\calT$ and are isolated in~$D'$. 
\end{claim2}
\begin{proof}
  In the last phase, note that cycles generated in the first phase will always
  contain a 2-edge.
  Therefore, the only possible increases in cost come from joining one of
  these~$c_p$ cycles.
  The increase is at most~2 if two such cycles are joined and at most~1 if one
  such cycle is joined to some cycle other than these~$c_p$ ones.
  So the increase in this phase is at most~$c_p$.
\hfill$\qed$
\end{proof}

\subsection{Approximation ratio}
\label{sub:sec:onetwo:119approx}

Theorem~\ref{thm:onetwo:119approx} shows how
Algorithm~\ref{alg:onetwo:119approx} guarantees an~11/9 approximation ratio
while Corollary~\ref{cor:onetwo:76approx} shows a case in which
Algorithm~\ref{alg:onetwo:119approx} can be adapted to guarantee a~7/6
approximation ratio.

\begin{theorem}
\label{thm:onetwo:119approx}
  Algorithm~\ref{alg:onetwo:119approx} is an $\frac{11}9$-approximation for the
  $\{1,2\}$-\steinercycle{} problem.
\end{theorem}
\begin{proof}
  Let~$(G,w,\calT)$ be an instance of the $\{1,2\}$-\steinercycle\ problem. 
  Let $n = |V(G)|$ and denote by $e_2(X)$ the total amount of 2-edges in a
  collection~$X$ of cycles.

  We start with two lower bounds on $\opt(G,w,\calT)$.
  Let $\calF$ be the 2-factor used in Algorithm~\ref{alg:onetwo:119approx} when
  applied to $(G,w,\calT)$.
  The first one is $w(\calF)$, because any solution for \steinercycle{}
  problem is a 2-factor in $G$.
  Thus 
  \begin{equation}
  \label{eq:onetwo:opt_lb1}
      \opt(G,w,\calT) \geq w(\calF) = n + e_2(\calF) \, .
  \end{equation}
  The other one is related to pure cycles in $\calF$.
  Consider an optimal solution~$\calC^*$ for instance $(G,w,\calT)$. 
  Thus $\opt(G,w,\calT) = n + e_2(\calC^*)$.
  Let $C^*_1, \ldots, C^*_r$ be the cycles of~$\calC^*$, where $C^*_i =
  (v_{i0}, \ldots, v_{i|C^*_i|})$ for each $i \in \{1,\ldots, r\}$, with
  $v_{i0} = v_{i|C^*_i|}$.  
  Let $U = \{v_{ij} \colon i \in \{1, \ldots, r\}, \, j \in \{0, \ldots,
  |C^*_i|-1\} \text{, and } v_{ij}v_{i\,j+1} \mbox{ is a 2-edge}\}$ and note
  that $|U| = e_2(\calC^*)$. 
  Let~$\ell$ be the number of pure cycles in the 2-factor $\calF$ that contain
  vertices in~$U$. 
  Clearly $e_2(\calC^*) \geq \ell$, which gives us
  \begin{equation}
  \label{eq:onetwo:opt_lb2}
      \opt(G,w,\calT) \geq n + \ell \, .
  \end{equation}
  
  Now let~$\calC$ be the 2-factor produced by
  Algorithm~\ref{alg:onetwo:119approx} for input $(G,w,\calT)$.
  Let us show an upper bound on the cost of~$\calC$.
  Solution~$\calC$ has cost~$w(\calF)$ plus the increase in the cost made in
  the first phase, and then in the final phase of joining cycles.
  Let us start bounding the total cost increase in the first phase. 
  Let~$c_p$ be as in Claim~\ref{clm:onetwo:phase2}.
  Recall that these~$c_p$ cycles are not matched by~$M$. 
  Let~$n(c_p)$ be the number of vertices in these~$c_p$ cycles, and note
  that~$n(c_p) \geq 3c_p$, because each such cycle does not respect $\calT$ and
  hence has at least three vertices.
  By Claim~\ref{clm:onetwo:phase1}, the vertices incident to 2-edges of~$\calF$ are
  never charged. 
  So there are at least~$e_2(\calF)$ vertices of the nonpure cycle of~$\calF$
  not charged during the first phase.
  Thus, at most $n - n(c_p) - e_2(\calF) \leq n - 3c_p - e_2(\calF)$ vertices
  were charged in the first phase.
  Also, by Claim~\ref{clm:onetwo:phase1}, each such vertex was charged at
  most~2/9.

  By Claim~\ref{clm:onetwo:phase2}, the increase in this phase is at most~$c_p$.
  Thus we have
  \begin{align}
      w(\calC) \ & \leq \ w(\calF) + \frac{2}{9}(n - 3c_p - e_2(\calF)) + c_p \nonumber \\
               & = \ n + e_2(\calF) + \frac{2}{9}(n - 3c_p - e_2(\calF)) + c_p \nonumber \\
               & = \ \frac{11}{9}n + \frac{7}{9}e_2(\calF) + \frac{1}{3}c_p \nonumber \\
               & \leq \ \frac{11}{9}n + \frac{7}{9}e_2(\calF) + \frac{1}{3}\ell \label{eq:onetwo:bound2} \\
               & \leq \ \frac{7}{9}(n + e_2(\calF)) + \frac{4}{9}(n + \ell) \nonumber \\
               & \leq \ \frac{7}{9}\,\opt(G,w,\calT) + \frac{4}{9}\,\opt(G,w,\calT) 
                 \ = \ \frac{11}{9}\,\opt(G,w,\calT) \, , \label{eq:onetwo:costalgorithm}
  \end{align}
  where~\eqref{eq:onetwo:bound2} holds by Claim~\ref{claim:onetwo:isolatedpure},
  and~\eqref{eq:onetwo:costalgorithm} holds by~\eqref{eq:onetwo:opt_lb1}
  and~\eqref{eq:onetwo:opt_lb2}.
  It remains to prove the following.

  \begin{claim2}
  \label{claim:onetwo:isolatedpure}
     $c_p \leq \ell$.
  \end{claim2}
  \begin{claimproof}
    Recall that $c_p$ is the number of pure cycles in~$\calF$ that are isolated
    in~$D'$ and do not respect~$\calT$, and observe that $\ell \le |U|$. 

    We will describe a matching in the bipartite graph~$B$ with at most~$\ell$
    unmatched cycles.
    From this, because~$M$ is a maximum matching in~$B$ and there are at
    least~$c_p$ cycles not matched by~$M$, we conclude that $c_p \leq \ell$.

    For each~$i \in \{1, \ldots, r\}$, go through the vertices of~$C^*_i$ from
    $j = 0$ to~$|C^*_i| - 1$ and if, for the first time, we find a vertex
    $v_{ij} \not\in U$ that belongs to a pure cycle~$C$ (which does not respect
    $\calT$) such that $v_{i\,j+1}$ is not in~$C$, we match~$C$ to~$v_{i\,j+1}$
    in~$B$.
    Note that, as $v_{ij} \not\in U$, the edge between~$C$ and $v_{i\,j+1}$ is
    indeed in~$B$.  
    Every pure cycle that does not respect~$\calT$ will be matched by this
    procedure, except for at most~$\ell$.
  \end{claimproof}
\hfill$\qed$
\end{proof}

This analysis is tight.  
Consider the instance depicted in Figure~\ref{fig:onetwo:119tight1}, with~$9$
vertices and $\calT = \{\{a_1,a_2,a_3\}, \{b_1,b_2,c_1,c_2,d_1,d_2\}\}$.  
There is a Hamiltonian cycle in the graph with only 1-edges, so the optimum
costs~$9$. 
However, there is also a 2-factor of cost~$9$ consisting of the three length-3
cycles~$C_1$, $C_2$ and~$C_3$, as in Figure~\ref{fig:onetwo:119tight1}. 
The matching in the graph~$B$ might correspond to the 1-edge between~$C_2$
and~$C_1$, and the 1-edge between~$C_3$ and~$C_2$, as in
Figure~\ref{fig:onetwo:119tight2}.
This leads to a length-2 path in~$D'$, as in Figure~\ref{fig:onetwo:119tight3}.
The process of joining these cycles, as the algorithm does, might lead to an
increase of~2 in the cost, resulting in the solution of cost~$11$ depicted in
Figure~\ref{fig:onetwo:119tight5}, which achieves a ratio of exactly~$11/9$. 
This example can be generalized to have $n=9k$ vertices, for any positive
integer~$k$.

\begin{figure}[ht!]
    \centering

    \begin{subfigure}[t]{0.47\textwidth}
        \centering
        \includegraphics[width=0.65\textwidth]{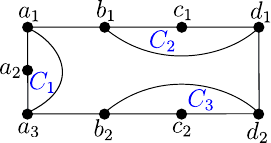}
        \caption{Initial graph~$G$ and 2-factor~$F$. All depicted edges are
        1-edges while the missing ones are 2-edges.}
        \label{fig:onetwo:119tight1}
    \end{subfigure}
    \hfill
    \begin{subfigure}[t]{0.47\textwidth}
        \centering
        \includegraphics[width=0.75\textwidth]{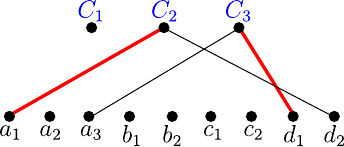}
        \caption{Bipartite graph built from~$F$ and matching~$M$ highlighted.}
        \label{fig:onetwo:119tight2}
    \end{subfigure}
    \medskip

    \begin{subfigure}[t]{0.22\textwidth}
        \centering
        \includegraphics[width=0.5\textwidth]{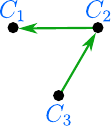}
        \caption{Digraph~$D$ which coincides with~$D'$.}
        \label{fig:onetwo:119tight3}
    \end{subfigure}
    \hfill
    \begin{subfigure}[t]{0.36\textwidth}
        \centering
        \includegraphics[width=\textwidth]{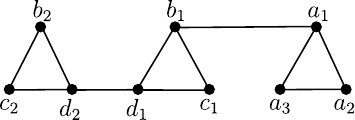}
        \caption{Joining cycles of the length-2 path in~$D'$.}
        \label{fig:onetwo:119tight4}
    \end{subfigure}
    \hfill
    \begin{subfigure}[t]{0.36\textwidth}
        \centering
        \includegraphics[width=\textwidth]{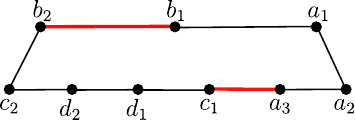}
        \caption{Final solution $\calC$, with two 2-edges, and cost 11.}
        \label{fig:onetwo:119tight5}
    \end{subfigure}

    \caption{Tight example for Algorithm~\ref{alg:onetwo:119approx}.}
    \label{fig:onetwo:119tight}
\end{figure}

Similarly to what Papadimitrou and
Yannakakis~\cite{1993-papadimitriou-yannakakis} achieve for the $\{1,2\}$-TSP,
we also derive the following.

\begin{corollary}
\label{cor:onetwo:76approx}
  Algorithm~\ref{alg:onetwo:119approx} is a $\frac76$-approximation for the
  $\{1,2\}$-\steinercycle{} problem when $|T| \geq 4$ for all $T \in \calT$.
\end{corollary}
\begin{proof}
  For weights 1 and 2, there is a polynomial-time algorithm that computes a
  minimum-weight 2-factor that contains no triangle~\cite[Section~3,
  Chapter~3]{Hartvigsen1984}.
  (See Appendix~\ref{appendix} for a discussion on references to this algorithm
  in the literature.)
  Using this algorithm within \Call{Special2Factor}{} in
  Algorithm~\ref{alg:onetwo:119approx}, we can guarantee that there are at least~4
  vertices per cycle in the produced 2-factor~$\calC$.
  The charging argument presented in Claim~\ref{clm:onetwo:phase1} can use the fact
  that the cycles have length at least~$4$, which increases the number of
  vertices to distribute the cost increase. 
  For instance, when we join a cycle~$C_i$ with~$C$, as in
  Figure~\ref{fig:onetwo:join-in-trees2}, the increase is at most~1, and we charge
  this cost to the vertices in~$C_i$ and also to the two vertices involved
  in~$C$.
  Now there are at least~4 vertices in~$C_i$, so each of these vertices is
  charged at most~1/6.  
  The other case in which the charged cost was more than~1/6 was when three
  cycles were joined, as in Figures~\ref{fig:onetwo:join-paths1}
  and~\ref{fig:onetwo:join-paths2}. 
  In this case, the extra cost is at most~2, which is now charged to the at
  least~12 vertices that belong to these cycles, giving a cost of at most~1/6
  per vertex.
  So the value charged per vertex is at most~1/6 in all cases, and the result
  follows.
\hfill$\qed$
\end{proof}

\section{Asymmetric Steiner Multicycle problem}
\label{sec:asymmetric}

In this section, we consider a version of the \steinercycle\ in which the input
graph is a complete digraph~$D$ on~$n$ vertices with arc set~$A(D)$, and the
weight function $w \colon A(D) \to \mathbb{Q}_+$ does not necessarily satisfy
$w(u,v) = w(v,u)$ for all $u,v \in V(D)$ with $u\neq v$.
We shall assume that the arc weights still satisfy the triangular inequality:
$w(a,c) \leq w(a,b)+w(b,c)$ for all distinct $a,b,c \in V(D)$.
As before, we also have a collection~$\calT$ of terminal sets which
partitions~$V(D)$, and the goal now is to find a minimum weight \emph{directed}
2-factor of~$D$ that respects~$\calT$.

We next devise an $\Oh(\lg n)$-approximation algorithm for this problem that is
inspired by the algorithm with the same approximation ratio for the Asymmetric
TSP, proposed by~Frieze, Galbiati, and Maffioli~\cite{Frieze1982}.
At each iteration, their algorithm proceeds as follows.
It starts with an induced subdigraph~$D'$ of~$D$ and what we call a
\emph{strongly Eulerian} spanning subdigraph~$\calC$ of~$D$ (initially $D' = D$
and~$\calC$ has no arcs).
Then, it finds a minimum weight 2-factor~$\calF$ in~$D'$, and makes $\calC =
\calC \cup \calF$.  
If~$\calF$ has only one cycle, then~$\calC$ is connected, and their algorithm
outputs a Hamiltonian cycle obtained from shortcutting~$\calC$ into a cycle. 
If~$\calF$ has more than one cycle, then their algorithm chooses a vertex in
each cycle of~$\calF$, called its \emph{representative}, it lets $D'$ be the
subdigraph of~$D$ induced on these representatives, and it starts the next
iteration with the new~$D'$ and~$\calC$.
The authors observed that each 2-factor~$\calC$ has weight bounded by the
length of the optimal TSP tour, and the number of iterations is bounded by $\lg
n$, because the number of components of~$\calC$ is divided by two in each
iteration. 
This implies the $\Oh(\lg n)$ approximation ratio.

Our algorithm aims at obtaining a 2-factor that respects~$\calT$.
Hence it stops once each terminal set is contained in a component of~$\calC$. 
It also differs from the algorithm due to Frieze, Galbiati, and
Maffioli~\cite{Frieze1982} in the way it chooses the representatives.
At each iteration of our algorithm, one has to guarantee that the
2-factor~$\calF$ has weight bounded by the optimal value, and that the number
of iterations is still~$\Oh(\lg n)$. 
We shall see that this can be done using a minimal edge cover of an auxiliary
graph to find \emph{good} representatives.
Recall that an \defi{edge cover} in a graph is a set~$M$ of edges such that
every vertex is incident to an edge in~$M$.  
A minimal edge cover on a graph with~$n$ vertices can be computed in time
$\Oh(n^{2.5})$ using the algorithm for the maximum matching problem in general
graphs due to Micali and Vazirani~\cite{micali1980v}.

A digraph~$D'$ is said to be \defi{strongly Eulerian} if, for every $v \in
V(D')$, the indegree and outdegree of~$v$ in~$D'$ are each equal to some $k(v)
\in \mathbb{Z}_+$, and~$D'-v$ contains precisely $k(v)-1$ more components
than~$D$.
We say that a component~$K$ of $D'-v$ is \defi{adjacent} to~$v$ if there is a
vertex in~$K$ which is a neighbor of~$v$ in~$D'$.
Analogously to the observation in~\cite{Frieze1982}, one may notice that, for
every $v \in V(D')$ and each connected component~$K$ of $D'-v$ which is
adjacent to~$v$, there exist distinct vertices $u,w \in V(K)$ such that $(u,v)$
and $(v,w)$ belong to~$A(D')$.
Procedure \textsc{DirectedShortcut} shows how to obtain a directed
2-factor~$\calF$ of~$D$ from a strongly Eulerian spanning subdigraph~$D'$
of~$D$ so that~$\calF$ has the same connected components as~$D'$. 
If there is an underlying weight function~$w \colon A(D) \to \mathbb{Q}_+$
satisfying the triangular inequality, then $w(\calF) \leq w(D')$.
For each iteration of the while loop in line~\ref{line:asy:while}, $D'$ is a
strongly Eulerian digraph with the same connected components.
We remark that this algorithm corresponds to the shortcutting procedure
described in~\cite{Frieze1982} applied to every component of~$D'$.
For the sake of completeness, procedure \textsc{DirectedShortcut} is presented
in Algorithm~\ref{alg:asy:dirShortCut}. 
Note that this takes polynomial time. 

\begin{algorithm}
  \begin{algorithmic}[1]
    \Require{a strongly Eulerian digraph $D'$}
    \Ensure{a directed 2-factor with the same connected components as $D'$}
    \While{$D'$ is not a directed 2-factor}\label{line:asy:while}
      \State Let $v \in V(D')$ be such that with $k(v) > 1$
      \State Let $C_1, C_2$ be distinct components of $D'-v$ that are adjacent to $v$
      \State Let $u_i,w_i \in V(C_i)$ such that $(u_i,v),(v,w_i) \in A(D)$ for $i \in \{1,2\}$
      \State $A(D') \gets [A(D') \setminus \{(u_1, v), (v,w_2)\}] \cup \{(u_1, w_2)\}$ \label{line:asy:digraph-shortcut}
    \EndWhile
    \State \Return $D'$
  \end{algorithmic}
  \caption{\textsc{DirectedShortcut}($D'$)}
  \label{alg:asy:dirShortCut}
\end{algorithm}

Let $\eta^{}_\calT(\calC)$ denote the number of cycles in a 2-factor~$\calC$ that do not
respect~$\calT$.
The procedure \textsc{Representatives} takes as input~$\calC$ and~$\calT$, and
it creates an auxiliary undirected graph~$G$ with vertex set being the
$\eta^{}_\calT(\calC)$ cycles in~$\calC$ that do not respect~$\calT$ and edge set 
$\{\{C,C'\} \colon C \neq C', V(C) \cap T \neq \emptyset \text{, and }
{V(C')\cap T \neq \emptyset}$ for some $T \in \calT\}$.
Then, it computes a minimal edge cover~$M$ of~$G$ and, for each edge $\{C,C'\}
\in M$, it chooses a pair of vertices $\{r^{}_C, r^{}_{C'}\}$ such that $r^{}_C
\in V(C) \cap T$ and $r^{}_{C'} \in V(C') \cap T$ where~$T$ is a terminal set
in~$\calT$ that intersects both~$C$ and~$C'$.
The procedure then returns the set of vertices $R = \bigcup_{\{C,C'\} \in M}
\{r^{}_C, r^{}_{C'}\}$.

\begin{algorithm}
  \begin{algorithmic}[1]
    \Require{a directed 2-factor $\calC$ of an induced subdigraph of $D$ and $\calT$}
    \Ensure{a set of vertices $R \subseteq V(D)$}

    \State Let $E = \{\{C,C'\} \colon C \neq C', V(C) \cap T \neq \emptyset \text{, and }
            V(C')\cap T \neq \emptyset \text{ for some } T \in \calT\}$ 
    \State Let $G$ be the graph with vertex set $\{C \in \calC : C \text{ does not respect } \calT\}$ and edge set $E$
    \State $M \gets \textsc{MinimalEdgeCover}(G)$
    \State $R \gets \emptyset$
    \For{each edge $\{C,C'\} \in M$ }
      \State Let $T \in \calT$ such that $T\cap C \neq \emptyset$ and $T\cap C' \neq \emptyset$
      \State Let $r^{}_C \in V(C) \cap T$ and $r^{}_{C'} \in V(C') \cap T$
      \State $R \gets R \cup \{r^{}_C, r^{}_{C'}\}$
    \EndFor
    \State \Return $R$
  \end{algorithmic}
  \caption{\textsc{Representatives}($\calC, \calT$)}
  \label{alg:asy:representatives}
\end{algorithm}

We next argue that Algorithm~\ref{alg:asy:representatives} produces a set~$R$
satisfying
\begin{enumerate}[(i)]
  \item $R \cap V(C) \neq \emptyset$ for every $C \in V(G)$;
  \item $|R \cap T| \neq 1$ for every terminal set $T \in \calT$; and
  \item $|R \cap V(C)| = 1$ for at least $\eta^{}_\calT(\calC)/2$ cycles~$C$ in $\calC$. \label{property:asy:lonely}
\end{enumerate}
The first property holds because~$M$ is an edge cover of~$G$, thus, for
every~$C \in V(G)$, at least one vertex from $V(C)$ was included in~$R$.
The second property follows from the fact that, for each edge in~$M$, two
distinct vertices of the same terminal set were simultaneously included in~$R$.
For a terminal set $T \in \calT$ contained in a cycle $C \in V(G)$, 
we have $|R \cap T| = 0$.
The last property holds because, in every minimal edge cover, at least half of
the vertices are covered exactly once.  
Indeed, every edge of a minimal edge cover is incident to a vertex that is only
covered by this edge, and there are at least $|V(G)|/2 = \eta^{}_\calT(\calC)/2$ 
edges in any edge cover of~$G$. 
Every cycle $C \in \calC$ such that $|R \cap V(C)| = 1$ is said to be
\defi{lonely}.
Note that property~\eqref{property:asy:lonely} guarantees that there are at least
$\eta^{}_\calT(\calC)/2$ lonely cycles.

Algorithm~\ref{alg:asy:lgnapprox} formalizes the steps of our algorithm for the
\textsc{Asymmetric \steinercycle}\ problem.
It uses an auxiliary procedure that computes a minimum weight directed 2-factor
in a weighted digraph.
See Figure~\ref{fig:asy:lgnapprox} for an example.

\begin{algorithm}
  \begin{algorithmic}[1]
    \Require{a complete digraph $D$, a weight function $w \colon A(D)
    \to \mathbb{Q}_+$, and a partition $\calT$ of~$V(D)$}
    \Ensure{a directed 2-factor~$\calC$ in $D$ that respects~$\calT$}
    
    \State $\calC \gets$ \Call{MinimumDirected2Factor}{$D$, $w$}\label{line:asy:C0}
    \While{$\eta^{}_\calT(\calC) > 0$}\label{line:asy:main-loop}
      \State $R \gets$ \Call{Representatives}{$\calC$, $\calT'$} \label{line:asy:R}
      \State Let $D'$ be the complete digraph induced by $R$ on~$D$
      \State Let $w'$ be the restriction of $w$ to $A(D')$
      \State $\calC' \gets$ \Call{MinimumDirected2Factor}{$D'$, $w'$} \label{line:asy:2factor}
      \State Let $D''$ be the digraph induced by $\calC' \cup \calC$ \label{line:asy:Eulerian-digraph}
      \State $\calC \gets $ \Call{DirectedShortcut}{$D''$} \label{line:asy:shortcut}
    \EndWhile
    \State \Return $\calC$
  \end{algorithmic}
  \caption{\textsc{SteinerMulticycleApprox\_Asymmetric}($G$, $w$, $\calT$)}
  \label{alg:asy:lgnapprox}
\end{algorithm}

\begin{figure}[ht!]
    \centering

    \begin{subfigure}{\textwidth}
        \centering
        \includegraphics[width=\textwidth]{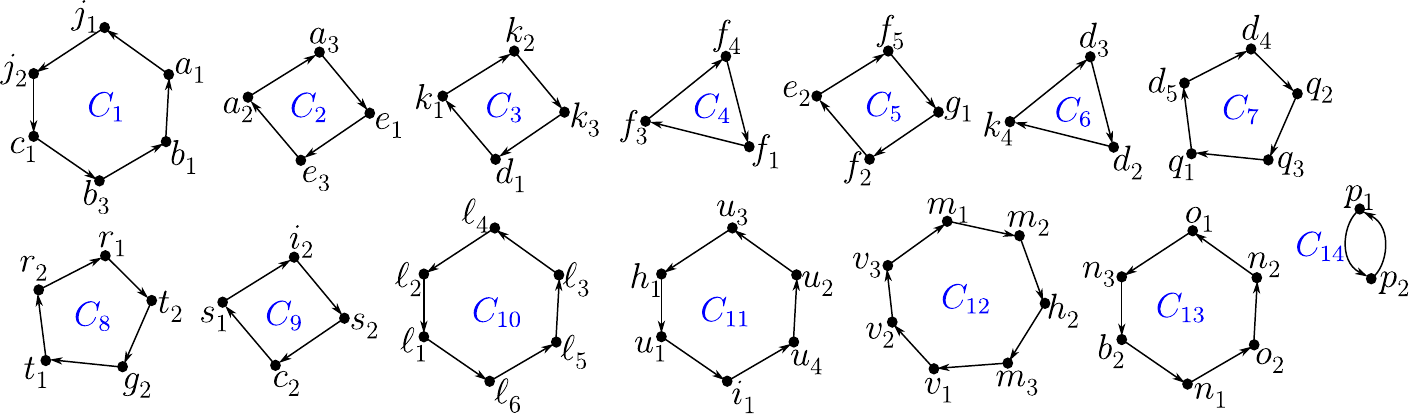}
        \caption{Minimum directed 2-factor $\calC$ obtained at line~\ref{line:asy:C0}.}
        \label{fig:asy:two_factor}
    \end{subfigure}

    \begin{subfigure}[t]{0.4\textwidth}
        \centering
        \includegraphics[width=0.65\textwidth]{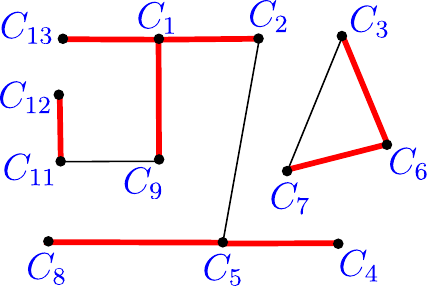}
        \caption{Edge cover obtained from $\calC$ for 
        $R=\{a_1,a_3,b_1,b_2,c_1,c_2,d_1,d_2,d_4,f_1,f_2,g_1,g_2,h_1,h_2\}$.}
        \label{fig:asy:edge_cover}
    \end{subfigure}
    \hfill
    \begin{subfigure}[t]{0.54\textwidth}
        \centering
        \includegraphics[width=\textwidth]{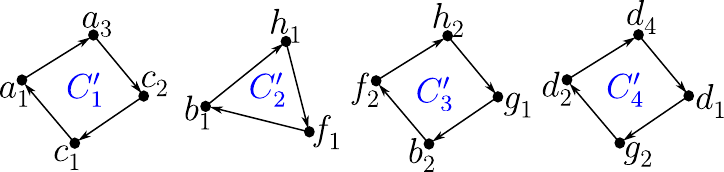}
        \caption{Minimum directed 2-factor $\calC'$ restricted to $D'$.}
        \label{fig:asy:two_factor_restricted}
    \end{subfigure}

    \begin{subfigure}{\textwidth}
        \centering
        \includegraphics[width=\textwidth]{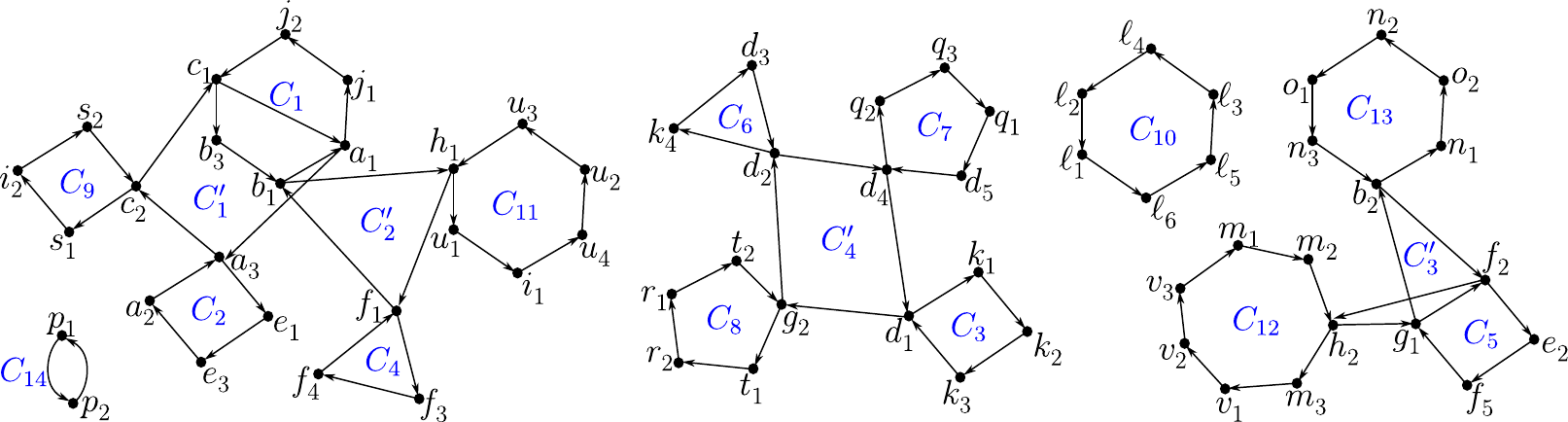}
        \caption{Strongly Eulerian digraph $D''$.}
        \label{fig:asy:strongly_eulerian}
    \end{subfigure}

    \begin{subfigure}{\textwidth}
        \centering
        \includegraphics[width=\textwidth]{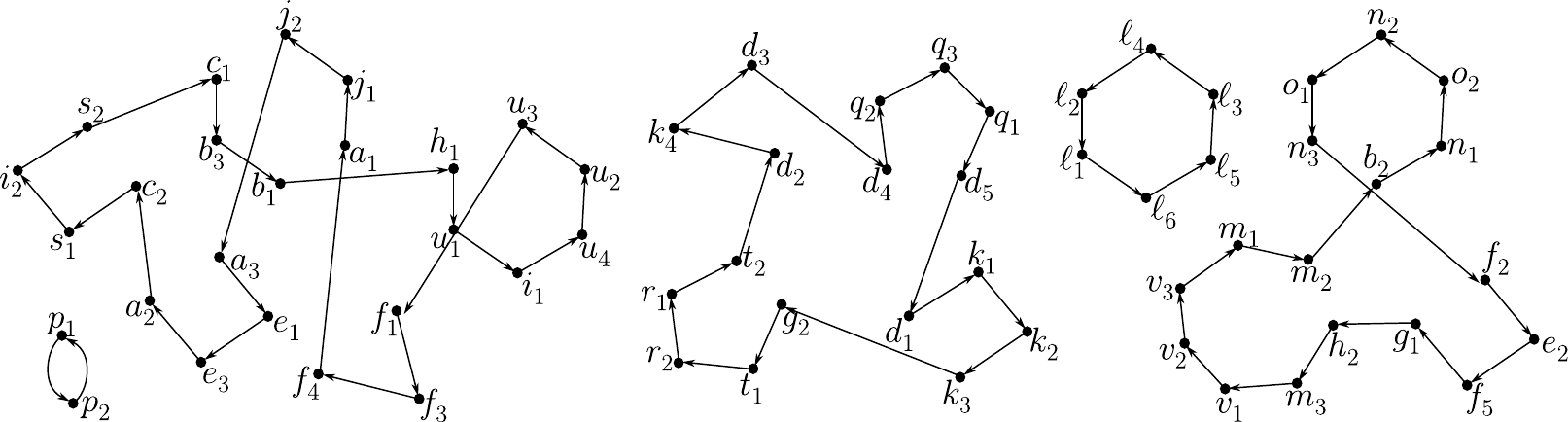}
        \caption{New 2-factor obtained from shortcutting an Eulerian tour in~$D''$.}
        \label{fig:asy:new_two_factorn}
    \end{subfigure}

    \caption{Auxiliar digraphs and structures built by Algorithm~\ref{alg:asy:lgnapprox}.}
    \label{fig:asy:lgnapprox}
\end{figure}

The way of choosing representatives in this algorithm is more complex than the
way used in~\cite{Frieze1982}.
This is because deriving an upper bound on the weight of~$\calC'$
in terms of an optimal 2-factor is more challenging than in terms of a minimum
weight TSP tour.
Specifically, a TSP tour can be shortcut into a 2-factor for any set~$R$ of
representatives.
However, this might not be the case for an optimal 2-factor.
Indeed, in~\cite{Frieze1982}, only one representative vertex is (arbitrarily)
chosen from each cycle.
However, in the example in Figure~\ref{fig:asy:lgnapprox}, suppose~$a_1$,
$a_2$, and~$a_3$ are alone in a cycle in any optimal 2-factor, and vertex~$a_1$
was chosen as the representative of~$C_1$, while vertex $e_1$ is chosen as the
representative of~$C_2$ (hence~$a_2$ and~$a_3$, which are also in~$C_2$, would
not be representatives). 
In this case, no shortcut of any optimal 2-factor would result in a 2-factor on
the chosen representatives: $a_1$ would be isolated in a shortcut of any optimal
2-factor on the chosen representatives.
This means we cannot guarantee that the optimal cost is an upper bound on the
minimum weight $w(\calC')$ of a 2-factor on the representatives. 
So we needed to develop a way to guarantee that the shortcut of an optimal
2-factor on~$R$ is a 2-factor, keeping the property that~$\calC'$ joins a good
amount of cycles of~$\calC$ that do not respect~$\calT$.
In the other extreme, one could consider including in~$R$ all vertices in
unhappy terminal sets because then the shortcut on~$R$ of any optimal solution
would be a 2-factor.
But this 2-factor might not join unhappy terminal sets: indeed, all terminal
sets might be unhappy in~$\calC$, and in this case~$R$ would be the whole set
of vertices and $\calC'=\calC$, leading the algorithm to loop forever.

\begin{theorem}
  Algorithm~\ref{alg:asy:lgnapprox} is an $\Oh(\lg n)$-approximation for the
  \textsc{Asymmetric} \steinercycle\ problem, where~$n$ is the number of
  vertices in the given digraph. 
\end{theorem}
\begin{proof} 
  Let $(D,w,\calT)$ be an instance of the \textsc{Asymmetric} \steinercycle,
  where~$D$ has~$n$ vertices. 
  We first show that the solution produced by
  Algorithm~\ref{alg:asy:lgnapprox} is indeed feasible for $(D,w,\calT)$. 
  It follows from its construction that the digraph $D''$ (computed at
  line~\ref{line:asy:Eulerian-digraph}) is a strongly Euclidean spanning subdigraph
  of~$D$, and so~$\calC$ is indeed a directed 2-factor in~$D$ with the same
  components as~$D''$.
  By the condition in line~\ref{line:asy:main-loop}, the set~$\calC$ returned by
  Algorithm~\ref{alg:asy:lgnapprox} respects~$\calT$, and thus~$\calC$ is a
  valid solution for $(D,w,\calT)$. 

  We now prove that the weight of each minimum weight 2-factor computed at
  line~\ref{line:asy:2factor} is upper bounded by the weight of an optimal solution
  for $(D,w,\calT)$. 
  Consider the complete digraph~$D'$ and the weight function~$w'$ used in
  line~\ref{line:asy:2factor}, and let~$\calC'$ be the minimum weight directed
  2-factor of~$D'$ obtained in line~\ref{line:asy:2factor}.
  Consider an optimal 2-factor~$\calC^*$ for the instance $(D,w,\calT)$,
  that is, a mininum weight 2-factor in $(D,w)$ that respects~$\calT$. 
  Now consider a shortcutting on~$\calC^*$ to go only through the
  vertices of~$R$, say $\calC^*_R$.
  Note that $\calC^*_R$ is certainly a 2-factor in $D'$ 
  because no cycle in $\calC^*_R$ has only one terminal in $R$.
  Also, $\calC^*_R$ has weight at most $\opt(D,w,\calT)$, 
  since~$w$ satisfies the triangular inequalities.
  As~$\calC'$ is a minimum 2-factor in $D'$, we have that 
  $w(\calC') \leq w(\calC^*_R)$, leading to $w(\calC') \leq \opt(D,w,\calT)$.

  Consider an iteration of the while loop in line~\ref{line:asy:main-loop}.
  Let $\calC$ be the directed 2-factor at the beginning of this iteration, $R$
  be the set from line~\ref{line:asy:R}, and $\hat \calC$ be the directed
  2-factor obtained in line~\ref{line:asy:shortcut}.
  The following assertion holds.

  \begin{claim2}
      $\eta^{}_\calT(\hat \calC) \le \frac{3}{4}\eta^{}_\calT(\calC)$.
  \end{claim2}
  \begin{claimproof}
    It suffices to argue that the difference $\Delta:= \eta^{}_\calT(\calC) -
    \eta^{}_\calT(\hat \calC)$ is at least half the number of lonely cycles
    in~$\calC$. 
    Let~$C$ be a lonely cycle in~$\calC$ and let~$C'$ be the cycle in~$\calC'$
    containing the single representative in $R \cap V(C)$. 
    If $C'$ contains a representative of a cycle that is not a lonely cycle
    in~$\calC$, then~$C$ contributes with~1 to~$\Delta$.
    Otherwise, every vertex in $C'$ is a representative of a lonely cycle
    in~$\calC$, and so~$C$ contributes with $(|C'|-1)/|C'|$ to~$\Delta$.
    As $|C'|\ge 2$, we conclude that every lonely cycle in~$\calC$ contributes
    with at least~$1/2$ to~$\Delta$.
    Because there are at least $\eta^{}_\calT(\calC)/2$ lonely cycles
    in~$\calC$ by property~\eqref{property:asy:lonely} of~$R$,
    we have $\eta^{}_\calT(\calC) - \eta^{}_\calT(\hat \calC) \geq
    \frac{1}{4}\eta^{}_\calT(\calC)$, which implies $\eta^{}_\calT(\hat \calC) \leq
    \eta^{}_\calT(\calC) - \frac{1}{4}\eta^{}_\calT(\calC) =
    \frac{3}{4}\eta^{}_\calT(\calC)$.
  \end{claimproof}

  As $\eta^{}_\calT(\calC_0) \leq n$ for the initial directed 2-factor~$\calC_0$
  from line~\ref{line:asy:C0}, it follows from the previous claim that the maximum
  number of iterations of the while loop in line~\ref{line:asy:main-loop} is
  $\Oh(\lg n)$.
  This implies that the weight of the solution produced by
  Algorithm~\ref{alg:asy:lgnapprox} is at most $\Oh(\lg n)\,\opt(D,w,\calT)$.
\hfill$\qed$
\end{proof}

\section{Final remarks}
\label{sec:final}

When there is only one terminal set, the \steinercycle\ turns into the TSP.
There is a $\frac32$-approximation for the metric TSP, so the first natural
question is whether there is also a $\frac32$-approximation for the metric
\steinercycle, or at least some approximation with a ratio better than~3. 

The difficulty in the Steiner forest is also a major difficulty in the
\steinercycle\ problem: how to find out what is the right way to cluster the
terminal sets. 
Indeed, if the number~$k$ of terminal sets is bounded by a constant, then one
can use brute force to guess the way an optimal solution clusters the terminal
sets, and then, in the case of the \steinercycle, apply any approximation for
the TSP to each instance induced by one of the clusters.  This leads to a
$\frac32$-approximation for any metric instance with bounded number of terminal
sets.
It also leads to better approximations for hereditary classes of instances for
which there are better approximations for the TSP. 

It would be nice to find out whether or not the cost of a minimum weight
perfect matching on the set of odd vertices of a minimum weight Steiner forest
is at most the optimum value for the \steinercycle. 

Observe that, for the $\{1,2\}$-\steinercycle, we can achieve the same
approximation ratio than the modified algorithm for the $\{1,2\}$-TSP, but for
the more general metric case, our ratio is twice the best ratio for the metric
TSP.
This comes from the fact that the backbone structure used in the solution for
the metric TSP (the MST and the minimum weight 2-factor) can be computed in
polynomial time.
For the $\{1,2\}$-\steinercycle\ we can still use the 2-factor, but the two
adaptations of the MST for the metric \steinercycle\ (the Steiner forest and
the survivable network design) are hard problems, for which we only have
2-approximations, not exact algorithms. 

In fact, for the $\{1,2\}$-TSP, better approximation algorithms are known:
there is an $\frac87$-approximation by Berman and Karpinski~\cite{BermanK2006},
and a $\frac76$-approximation and a faster $\frac87$-approximation by Adamaszek
et al.~\cite{AdamaszekMP2018}. 
The latter algorithms rely on some tools that we were not able to extend to the
$\{1,2\}$-\steinercycle.
On the other hand, the $\frac87$-approximation due to Berman and Karpinski
seems to be more amenable to an adaptation.

Recently, constant-factor approximations were presented for the asymmetric
TSP~\cite{SvenssonTV2020,TraubV2022}.  
Thus a natural direction for further research is to design constant-factor
approximation algorithms also for the Asymmetric \steinercycle{}.

\section*{Acknowledgements}
C. G. Fernandes was partially supported by the National Council for Scientific
and Technological Development -- CNPq (Proc.~310979/2020-0 and~423833/2018-9).
C. N. Lintzmayer was partially supported by CNPq (Proc.~312026/2021-8).
P. F. S. Moura was partially supported by the Fundação de Amparo à
Pesquisa do Estado de Minas Gerais -- FAPEMIG (APQ-01040-21).
This study was financed in part by the Coordenação de Aperfeiçoamento de
Pessoal de Nível Superior - Brasil (CAPES) - Finance Code 001, 
and by Grant \#2019/13364-7, São Paulo Research Foundation (FAPESP).

\section*{Competing interests}
The authors declare that they have no known competing financial interests or
personal relationships that could have appeared to influence the work reported
in this paper.

\section*{CRediT authorship contribution statement}
Cristina G. Fernandes: Conceptualization, Methodology, Validation, Writing.

\noindent Carla N. Lintzmayer: Conceptualization, Methodology, Validation, Writing.

\noindent Phablo F. S. Moura: Conceptualization, Methodology, Validation, Writing.

\bibliographystyle{splncs04}
\bibliography{bibfile}

\begin{thebibliography}{10}
\providecommand{\url}[1]{\texttt{#1}}
\providecommand{\urlprefix}{URL }
\providecommand{\doi}[1]{https://doi.org/#1}

\bibitem{AdamaszekMP2018}
Adamaszek, A., Mnich, M., Paluch, K.: New approximation algorithms for
  $(1,2)$-{TSP}. In: Chatzigiannakis, I., Kaklamanis, C., Marx, D., Sannella,
  D. (eds.) 45th International Colloquium on Automata, Languages, and
  Programming (ICALP 2018). Leibniz International Proceedings in Informatics
  (LIPIcs), vol.~107, pp. 9:1--9:14. Schloss Dagstuhl--Leibniz-Zentrum fuer
  Informatik, Dagstuhl, Germany (2018). \doi{10.4230/LIPIcs.ICALP.2018.9}

\bibitem{1998-arora}
Arora, S.: Polynomial time approximation schemes for {Euclidean Traveling
  Salesman} and other geometric problems. {Journal of the ACM}  \textbf{45}(5),
   753--782 (1998). \doi{10.1145/290179.290180}

\bibitem{BansalBT2009}
Bansal, N., Bravyi, S., Terhal, B.M.: Classical approximation schemes for the
  ground-state energy of quantum and classical {I}sing spin {H}amiltonians on
  planar graphs. Quantum Information \& Computation  \textbf{9}(7),  701--720
  (2009)

\bibitem{BermanK2006}
Berman, P., Karpinski, M.: $8/7$-approximation algorithm for $(1,2)$-{TSP}. In:
  Proc. of the 17th Annual ACM-SIAM Symposium on Discrete Algorithm (SODA). pp.
  641--648 (2006)

\bibitem{2015-borradaile-etal}
Borradaile, G., Klein, P.N., Mathieu, C.: A polynomial-time approximation
  scheme for {E}uclidean {S}teiner forest. {ACM Transactions on Algorithms}
  \textbf{11}(3),  19:1--19:20 (2015). \doi{10.1145/2629654}

\bibitem{Christofides76}
Christofides, N.: Worst-case analysis of a new heuristic for the traveling
  salesman problem. Technical Report~388, Carnegie Mellon University (1976)

\bibitem{Edmonds1965}
Edmonds, J.: {Paths, trees, and flowers}. {Canadian Journal of Mathematics}
  \textbf{17},  449--467 (1965). \doi{10.4153/CJM-1965-045-4}

\bibitem{EdmondsJ1973}
Edmonds, J., Johnson, E.L.: Matchings, {E}uler tours and the {C}hinese postman
  problem. Math. Programming  \textbf{5},  88--124 (1973)

\bibitem{Ergun2007a}
Ergun, O., Kuyzu, G., Savelsbergh, M.: Reducing truckload transportation costs
  through collaboration. Transportation Science  \textbf{41}(2),  206--221
  (2007). \doi{10.1287/trsc.1060.0169}

\bibitem{Ergun2007b}
Ergun, O., Kuyzu, G., Savelsbergh, M.: Shipper collaboration. Computers \&
  Operations Research  \textbf{34}(6),  1551--1560 (2007).
  \doi{10.1016/j.cor.2005.07.026}

\bibitem{2022-fernandes-etal}
Fernandes, C.G., Lintzmayer, C.N., Moura, P.F.S.: {Approximations
  for the Steiner Multicycle Problem}. In: Casta{\~{n}}eda, A.,
  Rodr{\'i}guez-Henr{\'i}quez, F. (eds.) {LATIN 2022: Theoretical Informatics}.
  pp. 188--203. {Springer International Publishing}, Cham (2022).
  \doi{10.1007/978-3-031-20624-5\_12}

\bibitem{Frieze1982}
Frieze, A.M., Galbiati, G., Maffioli, F.: On the worst-case performance of some
  algorithms for the asymmetric traveling salesman problem. Networks
  \textbf{12}(1),  23--39 (1982). \doi{10.1002/net.3230120103}

\bibitem{Hartvigsen1984}
Hartvigsen, D.: An extension of matching theory. Ph.D. thesis, Department of
  Mathematics, Carnegie Mellon University, Pittsburgh, PA, USA (1984),
  \url{https://david-hartvigsen.net/?page_id=33}

\bibitem{HartvigsenL2013}
Hartvigsen, D., Li, Y.: Polyhedron of triangle-free simple 2-matchings in
  subcubic graphs. Mathematical Programming  \textbf{138}(1--2),  43--82
  (2013). \doi{10.1007/s10107-012-0516-0}

\bibitem{Jain2001}
Jain, K.: A factor 2 approximation algorithm for the generalized {S}teiner
  network problem. Combinatorica  \textbf{21}(1),  39--60 (2001)

\bibitem{2020-lintzmayer-etal}
Lintzmayer, C.N., Miyazawa, F.K., Moura, P.F.S., Xavier, E.C.: Randomized
  approximation scheme for {Steiner Multi Cycle in the Euclidean} plane.
  {Theoretical Computer Science}  \textbf{835},  134--155 (2020).
  \doi{10.1016/j.tcs.2020.06.022}

\bibitem{LovaszP1986}
Lovász, L., Plummer, M.D.: Matching Theory, North-Holland Mathematics Studies,
  vol.~121. Elsevier (1986)

\bibitem{micali1980v}
Micali, S., Vazirani, V.V.: An $\mathcal{O}( \sqrt{|V|}| {E}|)$ algoithm for
  finding maximum matching in general graphs. In: 21st Annual Symposium on
  Foundations of Computer Science (SFCS 1980). pp. 17--27. IEEE (1980)

\bibitem{Orlin2013}
Orlin, J.B.: Max flows in $\mathrm{O}(nm)$ time, or better. In: Proc. of the
  45th Annual ACM Symposium on Theory of Computing (STOC). pp. 765--774 (2013).
  \doi{10.1145/2488608.2488705}

\bibitem{1993-papadimitriou-yannakakis}
Papadimitriou, C.H., Yannakakis, M.: The {T}raveling {S}alesman {P}roblem with
  distances one and two. {Mathematics of Operations Research}  \textbf{18}(1),
  1--11 (1993). \doi{10.1287/moor.18.1.1}

\bibitem{2018-pereira-etal}
Pereira, V.N.G., Felice, M.C.S., Hokama, P.H.D.B., Xavier, E.C.: The {S}teiner
  {M}ulti {C}ycle {P}roblem with applications to a collaborative truckload
  problem. In: {17th International Symposium on Experimental Algorithms
  (SEA'2018)}. pp. 26:1--26:13 (2018). \doi{10.4230/LIPIcs.SEA.2018.26}

\bibitem{RosenkrantzSL77}
Rosenkrantz, D.J., Stearns, R.E., Lewis, P.M.: An analysis of several
  heuristics for the traveling salesman problem. SIAM Journal on Computing
  \textbf{6},  563--581 (1977)

\bibitem{Salazar2003}
Salazar-Gonz{\'a}lez, J.J.: The {S}teiner cycle polytope. {European Journal of
  Operational Research}  \textbf{147}(3),  671--679 (2003).
  \doi{10.1016/S0377-2217(02)00359-4}

\bibitem{Schrijver2003}
Schrijver, A.: Combinatorial Optimization: Polyhedra and Efficiency.
  Springer-Verlag (2003)

\bibitem{SvenssonTV2020}
Svensson, O., Tarnawski, J., V\'{e}gh, L.A.: A constant-factor approximation
  algorithm for the asymmetric traveling salesman problem. Journal of the ACM
  \textbf{67}(6) (2020). \doi{10.1145/3424306}

\bibitem{TraubV2022}
Traub, V., Vygen, J.: An improved approximation algorithm for the asymmetric
  traveling salesman problem. SIAM Journal on Computing  \textbf{51}(1),
  139--173 (2022). \doi{10.1137/20M1339313}

\bibitem{Tutte1954}
Tutte, W.T.: A short proof of the factor theorem for finite graphs. Canadian
  Journal of Mathematics  \textbf{6},  347--352 (1954)

\bibitem{Vazirani2002}
Vazirani, V.V.: Approximation Algorithms. Springer (2002)

\end{thebibliography}

\appendix

\section{Minimum-weight triangle-free 2-factor}\label{appendix}

Hartvigsen, in his PhD thesis~\cite[Section~3, Chapter~3]{Hartvigsen1984},
described an algorithm that finds, in a given graph, a triangle-free simple
2-matching with the maximum number of edges.
In this appendix, we detail how to use his algorithm to find a minimum-weight
triangle-free 2-factor in a complete graph with all edge weights~1 or~2. 
Let us start by clarifying the notation involved, as it is used differently
throughout the literature.

Let~$H$ be a graph (not necessarily complete, and without weights). 
A subgraph of~$H$ whose maximum degree is~2 is sometimes called a
\emph{2-matching}, and it differs from a 2-factor as it allows for degree-1 and
degree-0 vertices. 
That is, a 2-matching is a collection of vertex-disjoint paths and cycles
in~$H$.

Sometimes, in the literature, a 2-matching is used to refer to a weight
function that assigns weight 0, 1, or 2 to each edge of a simple graph~$H$ so
that the sum of the weights of the edges incident to each vertex is at most~2. 
An edge that is assigned a weight of~2 works essentially as a length-2 cycle.
For this reason, sometimes in the literature, the 2-matching as we defined is
referred to as a \emph{simple} 2-matching (as it does not allow for these
parallel edges).
Also, a 2-factor is sometimes called a perfect simple 2-matching. 
Indeed, a simple 2-matching~$F$ is \emph{perfect} if every vertex is incident
to exactly two edges from~$F$. 

There are polynomial-time algorithms that find a minimum-weight 2-factor in a
complete graph with arbitrary edge weights.
Such an algorithm can be used to find a simple 2-matching in a given graph~$H$
with the maximum number of edges: just consider the edges of~$H$ as having
weight~1, and the non-edges as having weight~2, and throw away the weight-2
edges of the obtained 2-factor.

On the other hand, as far as we know, no polynomial-time algorithm is known to
find a minimum-weight \emph{triangle-free} 2-factor in a complete graph with
arbitrary edge weights. 
Indeed, Hartvigsen and Li~\cite{HartvigsenL2013} explicitly mention this as an
open problem.

There are some statements in the
literature~\cite{AdamaszekMP2018,1993-papadimitriou-yannakakis}, when
discussing the $\frac76$-approximation for TSP, that might lead one to think
that Hartvigsen's algorithm for finding a maximum-size triangle-free simple
2-matching could be used to find a minimum-weight triangle-free 2-factor for
general weights.
But that does not seem to be the case.
What is true, and stated explicitly in~\cite{HartvigsenL2013}, is that
Hartvigsen's algorithm can be used to find a \emph{maximum}-weight
triangle-free 2-factor in a complete graph~$G$ with edge weights~0 and~1.
For completeness, we detail how this can be achieved. 

One can apply the original algorithm of Hartvigsen~\cite[Section~3,
Chapter~3]{Hartvigsen1984} on the graph~$H$ obtained from~$G$ by removing all
edges of weight~0.
Hartvigsen's algorithm returns a triangle-free collection~$\calC$ of cycles and
paths in~$H$, and by joining the paths in~$\calC$ into a single cycle, using
edges of weight~0 (possibly an artificially and momentarily added loop or
parallel edge), we obtain a 2-factor whose weight is the number of edges
in~$\calC$.
If all cycles in this 2-factor have length at least~4, then we are done.
If not, then the only cycle~$C$ of length at most~3 is the cycle obtained from
joining the paths in~$\calC$.
If there is an appropriate weight-1 edge in~$G$ connecting~$C$ to one of the
other cycles in the 2-factor, then we can exchange a weight-1 and a weight-0
edge in the current 2-factor for this weight-1 edge and another weight-0 edge,
to obtain a triangle-free 2-factor in~$G$ with the same weight.  
If no appropriate weight-1 edge exists connecting~$C$ to the rest of the graph,
then we can do a similar exchange, but replacing a weight-1 and a weight-0 edge
with two weight-0 edges.
By a case analysis, one can verify that this leads to a maximum-weight
triangle-free 2-factor in~$G$.

In Corollary~\ref{cor:onetwo:76approx}, the given graph $G$ has edge weights 1 or~2, 
and we want to find a \emph{minimum}-weight triangle-free 2-factor in~$G$.
That can be solved similarly using Hartvigsen's algorithm on the graph 
with only the weight-1 edges. 

\end{document}